\documentclass[a4paper,11pt,british]{article}
\usepackage{babel}

\usepackage[left=1.5in,right=1.5in,top=1.5in,bottom=1.5in]{geometry}

\usepackage{graphicx}
\usepackage{amsmath}
\usepackage{array}
\usepackage{url}
\usepackage{amssymb}
\usepackage{amsfonts}
\usepackage{booktabs}
\usepackage{hyperref}
\usepackage{orcidlink}
\usepackage[caption=false]{subfig}

\usepackage[ruled,vlined]{algorithm2e}
\SetKw{KwPP}{Public Parameters:}
\usepackage{mathtools}
\DeclarePairedDelimiter\ceil{\lceil}{\rceil}
\DeclarePairedDelimiter\floor{\lfloor}{\rfloor}
\SetKwComment{Comment}{\scriptsize $\triangleright$\ }{}
\makeatletter\def\dftbig#1{{\hbox{$\left#1\vbox to10\p@{}\right.\n@space$}}}\makeatother
\renewcommand{\Pr}{\mathsf{Pr}}
\newcommand{\E}{\mathsf{E}}
\newcommand{\eat}[1]{}

\usepackage{amsthm}
\usepackage{thmtools}
\usepackage{thm-restate}
\newtheorem{thm}{Theorem}[section]
\newtheorem{lem}[thm]{Lemma}
\newtheorem{cor}[thm]{Corollary}
\theoremstyle{definition}

\begin{document}

\title{Applying the Shuffle Model of Differential Privacy to Vector Aggregation}

\author{Mary Scott\textsuperscript{\orcidlink{0000-0003-0799-5840}},
Graham Cormode\textsuperscript{\orcidlink{0000-0002-0698-0922}}, and
Carsten Maple\textsuperscript{\orcidlink{0000-0002-4715-212X}}}
\date{}

\maketitle

\begin{abstract}
In this work we introduce a new protocol for vector aggregation in the context of the Shuffle Model, a recent model within Differential Privacy (DP).
It sits between the Centralized Model, which prioritizes the level of accuracy over the secrecy of the data, and the Local Model, for which an improvement in trust is counteracted by a much higher noise requirement.
The Shuffle Model was developed to provide a good balance between these two models through the addition of a shuffling step, which unbinds the users from their data whilst maintaining a moderate noise requirement.
We provide a single message protocol for the summation of real vectors in the Shuffle Model, using advanced composition results.
Our contribution provides a mechanism to enable private aggregation and analysis across more sophisticated structures such as matrices and higher-dimensional tensors, both of which are reliant on the functionality of the vector case.
\end{abstract}

\section{Introduction}  \label{sec:intro}
Differential Privacy (DP)~\cite{dworkintro} is a strong, mathematical definition of privacy that guarantees a measurable level of confidentiality for any data subject in the dataset to which it is applied.
In this way, useful collective information can be learned about a population, whilst simultaneously protecting the personal information of each data subject.

In particular, DP guarantees that the impact on any particular individual as a result of analysis on a dataset is the same, whether or not the individual is included in the dataset.
This guarantee is quantified by a parameter $\varepsilon$, which represents good privacy if it is small.
However, finding an algorithm that achieves DP often requires a trade-off between privacy and accuracy, as a smaller $\varepsilon$ sacrifices accuracy for better privacy, and vice versa.
DP enables data analyses such as the statistical analysis of the salaries of a population.
This allows useful collective information to be studied, as long as $\varepsilon$ is adjusted appropriately to satisfy the definition of DP.

In this work we focus on protocols in the \emph{Single-Message Shuffle Model}~\cite{balleprivacyblanket}, a one-time data collection model where each of $n$ users is permitted to submit a single message.
We have chosen to apply the Single-Message Shuffle Model to the problem of \emph{vector aggregation}, as there are links to Federated Learning and Secure Aggregation.

There are many practical applications of the Single-Message Shuffle Model in Federated Learning, where multiple users collaboratively solve a Machine Learning problem, the results of which simultaneously improves the model for the next round~\cite{mcmahan}.
The updates generated by the users after each round are high-dimensional vectors, so this data type will prove useful in applications such as training a Deep Neural Network to predict the next word that a user types~\cite{neural}.
Additionally, aggregation is closely related to Secure Aggregation, which can be used to compute the outputs of Machine Learning problems such as the one above~\cite{bonawitz}.

Our contribution is a protocol in the Single-Message Shuffle Model for the private summation of vector-valued messages, extending an existing result from Balle \emph{et al.}~\cite{balleprivacyblanket} by permitting the $n$ users to each submit a vector of real numbers instead of a scalar.
The resulting estimator is unbiased and has normalized mean squared error (MSE) $O_{\varepsilon, \delta} (d^{8/3} n^{-5/3})$, where $d$ is the dimension of each vector.

This vector summation protocol above can be extended to produce a similar protocol for the linearization of matrices. It is important to use matrix reduction to ensure that the constituent vectors are linearly independent. This problem can be extended further to higher-dimensional tensors, which are useful for the representation of multi-dimensional data in Neural Networks.

\section{Related Work} \label{sec:litreview}

The earliest attempts at protecting the privacy of users in a dataset focused on simple ways of suppressing or generalising the data.
Examples include $k$-anonymity~\cite{kanonymity}, $l$-diversity~\cite{ldiversity} and $t$-closeness~\cite{tcloseness}.
However, such attempts have been shown to be insufficient, as proved by numerous examples~\cite{dwork}.

This harmful leakage of sensitive information can be easily prevented through the use of DP, as this mathematically guarantees that the chance of a \emph{linkage attack} on an individual in the dataset is almost identical to that on an individual not in the dataset.

Ever since DP was first conceptualized in 2006 by Dwork \emph{et al.}~\cite{dworkintro}, the majority of research in the field has focused on two opposing models.
In the Centralized Model, users submit their sensitive personal information directly to a \emph{trusted} central data collector, who adds \emph{random noise} to the raw data to provide DP, before assembling and analyzing the aggregated results.

In the Local Model, DP is guaranteed when each user applies a \emph{local randomizer} to add random noise to their data before it is submitted. 
The Local Model differs from the Centralized Model in that the central entity does not see the users' raw data at any point, and therefore does not have to be trusted.
However, the level of noise required per user for the same privacy guarantee is much higher, which limits the usage of Local Differential Privacy (LDP) to major companies such as Google~\cite{google}, Apple~\cite{apple} and Microsoft~\cite{microsoft}.

Neither of these two extensively studied models can provide a good balance between the trust of the central entity and the level of noise required to guarantee DP. Hence, in recent years researchers have tried to create intermediate models that reap the benefits of both.

In 2017, Bittau \emph{et al.}~\cite{bittau} introduced the Encode, Shuffle, Analyze (ESA) model, which provides a general framework for the addition of a \emph{shuffling step} in a private protocol.
After the data from each user is encoded, it is randomly permuted to unbind each user from their data before analysis takes place.
In 2019, Cheu \emph{et al.}~\cite{cheu} formalized the Shuffle Model as a special case of the ESA model, which connects this additional shuffling step to the Local Model.
In the Shuffle Model, the local randomizer applies a randomized mechanism on a per-element basis, potentially replacing a truthful value with another randomly selected domain element.
The role of these independent reports is to create what we call a \emph{privacy blanket}, which masks the outputs which are reported truthfully.

As well as the result on the private summation of scalar-valued messages in the Single-Message Shuffle Model that we will be using~\cite{balleprivacyblanket}, Balle \emph{et al.} have published two more recent works that solve related problems.
The first paper~\cite{balleimproved} improved the distributed $n$-party summation protocol from Ishai \emph{et al.}~\cite{ishai} in the context of the Single-Message Shuffle Model to require $O (1 + \pi/\log n)$ scalar-valued messages, instead of a logarithmic dependency of $O (\log n + \pi)$, to achieve statistical security $2^{-\pi}$.
The second paper~\cite{ballemulti} introduced two new protocols for the private summation of scalar-valued messages in the Multi-Message Shuffle Model, an extension of the Single-Message Shuffle Model that permits each of the $n$ users to submit more than one message, using several independent shufflers to securely compute the sum.
In this work, Balle \emph{et al.} contributed a recursive construction based on the protocol from~\cite{balleprivacyblanket}, as well as an alternative mechanism which implements a discretized distributed noise addition technique using the result from Ishai \emph{et al.}~\cite{ishai}.

Also relevant to our research is the work of Ghazi \emph{et al.}~\cite{ghazipower}, which explored the related problems of private frequency estimation and selection in a similar context, drawing comparisons between the errors achieved in the Single-Message and Multi-Message Shuffle Models.
A similar team of authors produced a follow-up paper~\cite{ghazidistributed} describing a more efficient protocol for private summation in the Single-Message Shuffle Model, using the `invisibility cloak' technique to facilitate the addition of zero-sum noise without coordination between the users.

\section{Preliminaries} \label{sec:prelim}

We consider randomized mechanisms~\cite{dwork} $\mathcal{M}$, $\mathcal{R}$ under domains $\mathbb{X}$, $\mathbb{Y}$, and apply them to input datasets $\vec{D}, \vec{D}'$ to generate (vector-valued) messages $\vec{x}_{i}, \vec{x}_{i}'$. We write $[k] = \{ 1, \dots, k \}$ and $\mathbb{N}$ for the set of natural numbers.

\subsection{Models of Differential Privacy} \label{sec:dp}

The essence of Differential Privacy (DP) is the requirement that the contribution $\vec{x}_{i}$ of a user $i$ to a dataset $\vec{D} = (\vec{x}_{1}, \dots, \vec{x}_{n})$ does not have much effect on the outcome of the mechanism applied to that dataset.

In the \emph{centralized} model of DP, random noise is only introduced after the users' inputs are gathered by a (trusted) aggregator.
Consider a dataset $\vec{D}'$ that differs from $\vec{D}$ only in the contribution of a single user, denoted $\vec{D} \simeq \vec{D}'$.
Also let $\varepsilon \geq 0$ and $\delta \in (0,1)$. We say that a randomized mechanism $\mathcal{M}: \mathbb{X}^{n} \rightarrow \mathbb{Y}$ is $(\varepsilon, \delta)$-differentially private if $\forall \vec{D} \simeq \vec{D}', \forall E \subseteq \mathbb{Y}$:
\[
\Pr[\mathcal{M}(\vec{D}) \in E] \leq e^{\varepsilon} \cdot \Pr[\mathcal{M}(\vec{D}') \in E] + \delta~\cite{dwork}.
\]

In this definition, we assume that the trusted aggregator obtains the raw data from all users and introduces the necessary perturbations.

In the \emph{local} model of DP, each user $i$ independently uses randomness on their input $\vec{x}_{i} \in \mathbb{X}$ by using a \emph{local randomizer} $\mathcal{R}: \mathbb{X} \rightarrow \mathbb{Y}$ to obtain a perturbed result $\mathcal{R}(\vec{x}_{i})$. We say that the local randomizer is $(\varepsilon, \delta)$-differentially private if $\forall \vec{D}, \vec{D}', \forall E \subseteq \mathbb{Y}$:
\[
\Pr[\mathcal{R}(\vec{x}_{i}) \in E] \leq e^{\varepsilon} \cdot \Pr[\mathcal{R}(\vec{x}_{i}') \in E] + \delta~\cite{balleprivacyblanket}, 
\]

\noindent where $\vec{x}_{i}' \in \mathbb{X}$ is some other valid input vector that $i$ could hold.
The Local Model guarantees that any observer will not have access to the raw data from any of the users.
That is, it removes the requirement for trust. The price is that this requires a higher level of noise per user to achieve the same privacy guarantee.

\subsection{Single-Message Shuffle Model} \label{sec:smsm}

The Single-Message Shuffle Model sits in between the Centralized and Local Models of DP~\cite{balleprivacyblanket}.
Let a protocol $\mathcal{P}$ in the Single-Message Shuffle Model be of the form $\mathcal{P} = (\mathcal{R}, \mathcal{A})$, where $\mathcal{R}: \mathbb{X} \rightarrow \mathbb{Y}$ is the \emph{local randomizer}, and $\mathcal{A}: \mathbb{Y}^{n} \rightarrow \mathbb{Z}$ is the \emph{analyzer} of $\mathcal{P}$. 
Overall, $\mathcal{P}$ implements a mechanism $\mathcal{P}: \mathbb{X}^{n} \rightarrow \mathbb{Z}$ as follows. Each user $i$ independently applies the local randomizer to their message $\vec{x}_{i}$ to obtain a message $\vec{y}_{i} = \mathcal{R}(\vec{x}_{i})$.
Subsequently, the messages $(\vec{y}_{1}, \dots, \vec{y}_{n})$ are randomly permuted by a trusted \emph{shuffler} $\mathcal{S}: \mathbb{Y}^{n} \rightarrow \mathbb{Y}^{n}$.
The random permutation $\mathcal{S}(\vec{y}_{1}, \dots, \vec{y}_{n})$ is submitted to an untrusted data collector, who applies the analyzer $\mathcal{A}$ to obtain an output for the mechanism.
In summary, the output of $\mathcal{P}(\vec{x}_{1}, \dots, \vec{x}_{n})$ is given by:
\[
\mathcal{A} \circ \mathcal{S} \circ \mathcal{R}^{n}(\vec{x}) = \mathcal{A}(\mathcal{S}(\mathcal{R}(\vec{x}_{1}), \dots, \mathcal{R}(\vec{x}_{n}))).
\]

Note that the data collector observing the shuffled messages $\mathcal{S}(\vec{y}_{1}, \dots, \vec{y}_{n})$ obtains no information about which user generated each of the messages.
Therefore, the privacy of $\mathcal{P}$ relies on the indistinguishability between the shuffles $\mathcal{S} \circ \mathcal{R}^{n}(\vec{D})$ and $\mathcal{S} \circ \mathcal{R}^{n}(\vec{D}')$ for datasets $\vec{D} \simeq \vec{D}'$.
The analyzer can represent the shuffled messages as a \emph{histogram}, which counts the number of occurrences of the possible outputs of $\mathbb{Y}$.

\subsection{Measuring Accuracy} \label{sec:mse}

In Section~\ref{sec:vectorsum} we use the \emph{mean squared error} to compare the overall output of a private summation protocol in the Single-Message Shuffle Model with the original dataset.
The MSE is used to measure the average squared difference in the comparison between a fixed input $f(\vec{D})$ to the randomized protocol $\mathcal{P}$, and its output $\mathcal{P}(\vec{D})$.
In this context, $\text{MSE}(\mathcal{P}, \vec{D}) = \E \dftbig[ (\mathcal{P}(\vec{D}) - f(\vec{D}))^{2} \dftbig]$, where the expectation is taken over the randomness of $\mathcal{P}$. 
Note when $\E[\mathcal{P}(\vec{D})] =  f(\vec{D})$, MSE is equivalent to variance, i.e.:
\[
\text{MSE}(\mathcal{P}, \vec{D}) = \E \dftbig[ (\mathcal{P}(\vec{D}) - \E[\mathcal{P}(\vec{D})])^{2} \dftbig] = \text{Var}[\mathcal{P}(\vec{D})].
\]

\section{Vector Sum in the Shuffle Model} \label{sec:vectorsum}

In this section we introduce our protocol for vector summation in the Shuffle Model and tune its parameters to optimize accuracy.

\subsection{Basic Randomizer} \label{sec:basic}

First, we describe a basic local randomizer applied by each user $i$ to an input $x_i \in [k]$.
The output of this protocol is a (private) histogram of shuffled messages over the domain $[k]$.

The Local Randomizer $\mathcal{R}_{\gamma, k, n}^{PH}$, shown in Algorithm~\ref{alg:localrand}, applies a generalized \emph{randomized response} mechanism that returns the true message $x_i$ with probability $1 - \gamma$ and a uniformly random message with probability $\gamma$.
Such a basic randomizer is used by Balle \emph{et al.}~\cite{balleprivacyblanket} in the Single-Message Shuffle Model for scalar-valued messages, as well as in several other previous works in the Local Model~\cite{kairouzextremal, kairouzdiscrete, bhowmick}.
In Section~\ref{sec:privanalysis}, we find an appropriate $\gamma$ to optimize the proportion of random messages that are submitted, and therefore guarantee DP.

\begin{algorithm}[t]
\DontPrintSemicolon
\SetArgSty{textnormal}
\KwPP{\parbox[t]{2.1in}{\raggedright $\gamma \in [0,1]$, domain size $k$, and number of parties $n$}}\\
\KwIn{$x_i \in [k]$}
\KwOut{$y_i \in [k]$}
Sample $b \leftarrow$ {\tt Ber}$(\gamma)$\\
\lIf{$b=0$}{let $y_i \leftarrow x_i$}
\lElse{sample $y_i \leftarrow$ {\tt Unif}$([k])$}
\KwRet{$y_i$}
\caption{Local Randomizer $\mathcal{R}_{\gamma, k, n}^{PH}$}
\label{alg:localrand}
\end{algorithm}

We now describe how the presence of these random messages can form a `privacy blanket' to protect against a \emph{difference attack} on a particular user. 
Suppose we apply Algorithm~\ref{alg:localrand} to the messages from all $n$ users.
Note that a subset $B$ of approximately $\gamma n$ of these users returned a uniformly random message, while the remaining users returned their true message. 
Following Balle \emph{et al.}~\cite{balleprivacyblanket}, the analyzer can represent the messages sent by users in $B$ by a histogram $Y_1$ of uniformly random messages, and can form a histogram $Y_2$ of truthful messages from users not in $B$. 
As these subsets are mutually exclusive and collectively exhaustive, the information represented by the analyzer is equivalent to the histogram $Y = Y_1 \cup Y_2$.

Consider two neighbouring datasets, each consisting of $n$ messages from $n$ users, that differ only on the input from the $n^{\text{th}}$ user.
To simplify the discussion and subsequent proof, we temporarily omit the action of the shuffler.
By the post-processing property of DP, this can be reintroduced later on without adversely affecting the privacy guarantees.
To achieve DP we need to find an appropriate $\gamma$ such that when Algorithm~\ref{alg:localrand} is applied, the change in $Y$ is appropriately bounded.
As the knowledge of either the set $B$ or the messages from the first $n - 1$ users does not affect DP, we can assume that the analyzer knows both of these details.
This lets the analyzer remove all of the truthful messages associated with the first $n - 1$ users from $Y$.

If the $n^{\text{th}}$ user is in $B$, this means their submission is independent of their input, so we trivially satisfy DP.
Otherwise, the (curious) analyzer knows that the $n^{\text{th}}$ user has submitted their true message $x_n$.
The analyzer can remove all of the truthful messages associated with the first $n-1$ users from $Y$, and obtain $Y_1 \cup \{ x_n \}$.
The subsequent privacy analysis will argue that this does not reveal $x_n$ if $\gamma$ is set so that $Y_1$, the histogram of random messages, appropriately `hides' $x_n$.

\subsection{Private Summation of Vectors} \label{sec:protocol}

Here, we extend the protocol from Section~\ref{sec:basic} to address the problem of computing the sum of $n$ real vectors, each of the form $\vec{x}_{i} = (x_{i}^{(1)}, \dots, x_{i}^{(d)}) \in [0,1]^{d}$, in the Single-Message Shuffle Model.
Specifically, we analyze the utility of a protocol $\mathcal{P}_{d, k, n, t} = (\mathcal{R}_{d, k, n, t}, \mathcal{A}_{d, k, t})$ for this purpose, by using the MSE from Section~\ref{sec:mse} as the accuracy measure.
In the scalar case, each user applies the protocol to their entire input~\cite{balleprivacyblanket}.
Moving to the vector case, we allow each user to independently sample a set of $1 \le t \le d$ coordinates from their vector to report.
Our analysis allows us to optimize the parameter $t$.

Hence, the first step of the Local Randomizer $\mathcal{R}_{d, k, n, t}$, presented in Algorithm~\ref{alg:fixedpoint}, is to uniformly sample $t$ coordinates $(\alpha_{1}, \dots, \alpha_{t}) \in [d]$ (without replacement) from each vector $\vec{x}_{i}$.
To compute a differentially private approximation of $\sum_{i} \vec{x}_{i}$, we fix a quantization level $k$.
Then we randomly round each $x_{i}^{(\alpha_{j})}$ to obtain $\bar{x}_{i}^{(\alpha_j)}$ as either $\floor{{x}_{i}^{(\alpha_{j})} k}$ or $\ceil{{x}_{i}^{(\alpha_{j})} k}$.
Next, we apply the randomized response mechanism from Algorithm~\ref{alg:localrand} to each $\bar{x}_{i}^{(\alpha_{j})}$, which sets each output $y_{i}^{(\alpha_{j})}$ independently to be equal to $\bar{x}_{i}^{(\alpha_{j})}$ with probability $1 - \gamma$, or a random value in $\{ 0, 1, \dots, k \}$ with probability $\gamma$.
Each $y_{i}^{(\alpha_{j})}$ will contribute to a histogram of the form $(y_{1}^{(\alpha_{j})}, \dots, y_{n}^{(\alpha_{j})})$ as in Section~\ref{sec:basic}.

The Analyzer $\mathcal{A}_{d, k, t}$, shown in Algorithm~\ref{alg:analyzer}, aggregates the histograms to approximate $\sum_{i} \vec{x}_{i}$ by post-processing the vectors coordinate-wise.
More precisely, the analyzer sets each output $y_{i}^{(\alpha_{j})}$ to $y_{i}^{(l)}$, where the new label $l$ is from its corresponding input $x_{i}^{(l)}$ of the original $d$-dimensional vector $\vec{x}_{i}$.
For all inputs $x_{i}^{(l)}$ that were not sampled, we set $y_{i}^{(l)} = 0$.
Subsequently, the analyzer aggregates the sets of outputs from all users corresponding to each of those $l$ coordinates in turn, so that a $d$-dimensional vector is formed.
Finally, a standard debiasing step is applied to this vector to remove the scaling and rounding applied to each submission.
\texttt{DeBias} returns an unbiased estimator, $\vec{z}$, which calculates an estimate of the true sum of the vectors by subtracting the expected uniform noise from the randomized sum of the vectors.

\begin{algorithm}[t]
\DontPrintSemicolon
\SetArgSty{textnormal}
\KwPP{$k$, $t$, dimension $d$, and number of parties $n$}\\
\KwIn{$\vec{x}_{i} = (x_{i}^{(1)}, \dots, x_{i}^{(d)}) \in [0,1]^{d}$}
\KwOut{$\vec{y}_{i} = (y_{i}^{(\alpha_{1})}, \dots, y_{i}^{(\alpha_{t})}) \in \{0,1, \dots, k\}^{t}$}
\medskip
Sample $(\alpha_{1}, \dots, \alpha_{t}) \leftarrow$ {\tt Unif}$([d])$\\
Let $\bar{x}_{i}^{(\alpha_{j})} \leftarrow \floor[\dftbig]{x_{i}^{(\alpha_{j})} k} \ +$ {\tt Ber}$ (x_{i}^{(\alpha_{j})} k - \floor[\dftbig]{x_{i}^{(\alpha_{j})} k})$
\Comment*[r]{\parbox[t]{4in}{\raggedright $\bar{x}_{i}^{(\alpha_{j})}$: encoding of $x_{i}^{(\alpha_{j})}$ with precision $k$}}
\vspace*{-.3cm}
\Comment*[r]{\parbox[t]{5in}{\raggedright $y_{i}^{(\alpha_{j})}$: apply \textbf{Algorithm 1} to each $\bar{x}_{i}^{(\alpha_{j})}$}}
\KwRet{$\vec{y}_{i} = (y_{i}^{(\alpha_{1})}, \dots, y_{i}^{(\alpha_{t})})$}
\caption{Local Randomizer $\mathcal{R}_{d, k, n, t}$}
\label{alg:fixedpoint}
\end{algorithm}

\begin{algorithm}[t]
\DontPrintSemicolon
\SetArgSty{textnormal}
\KwPP{$k$, $t$, and dimension $d$}\\
\KwIn{Multiset $\bigl\{ \vec{y}_{i} \bigr\}_{i \in [n]}$, with $(y_{i}^{(\alpha_{1})}, \dots, y_{i}^{(\alpha_{t})}) \in \{0,1, \dots, k\}^{t}$}
\KwOut{$\vec{z} = (z^{(1)}, \dots, z^{(d)}) \in [0,1]^{d}$}
\medskip
Let $y_{i}^{(l)} \leftarrow y_{i}^{(\alpha_{j})}$
\medskip
\Comment*[r]{\parbox[t]{4.3in}{\raggedright $y_{i}^{(\alpha_{j})}$: submission corresponding to $x_{i}^{(l)}$}}
Let $(\hat{z}^{(1)}, \dots, \hat{z}^{(d)}) \leftarrow (\frac{1}{k} \sum_{i} y_{i}^{(1)}, \dots, \frac{1}{k} \sum_{i} y_{i}^{(d)})$\\
Let $(z^{(1)}, \dots, z^{(d)}) \leftarrow (${\tt DeBias}$ (\hat{z}^{(1)}), \dots, ${\tt DeBias}$ (\hat{z}^{(d)}))$
\Comment*[r]{\parbox[t]{4in}{\raggedright {\tt DeBias}$(\hat{z}^{(l)}) = ( \hat{z}^{(l)} - \frac{\gamma}{2} \cdot | y_{i}^{(l)}|) / (1 - \gamma)$}}
\KwRet{$\vec{z} = (z^{(1)}, \dots, z^{(d)})$}
\caption{Analyzer $\mathcal{A}_{d, k, t}$}
\label{alg:analyzer}
\end{algorithm}

\subsection{Privacy Analysis of Algorithms} \label{sec:privanalysis}

In this section, we will find an appropriate $\gamma$ that ensures that the mechanism described in Algorithms~\ref{alg:fixedpoint} and \ref{alg:analyzer} satisfies $(\varepsilon, \delta)$-DP for vector-valued messages in the Single-Message Shuffle Model.
To achieve this, we prove the following theorem, where we initially assume $\varepsilon < 1$ to simplify our computations.
At the end of this section, we discuss how to cover the additional case $1 \leq \varepsilon < 6$ to suit our experimental study.

\begin{thm} \label{thm:main}
The shuffled mechanism $\mathcal{M} = \mathcal{S} \circ \mathcal{R}_{d, k, n, t}$ is $(\varepsilon, \delta)$-DP for any $d, k, n \in \mathbb{N}$, $\{ t \in \mathbb{N} \ | \ t \in [d] \}$, $\varepsilon < 6$ and $\delta \in \left( 0,1 \right]$ such that:
\[
\gamma = 
\begin{cases}
\frac{56dk \log(1/\delta) \log(2t/\delta)}{(n - 1) \varepsilon^{2}}, & \text{when}\ \varepsilon < 1 \\
\frac{2016dk \log(1/\delta) \log(2t/\delta)}{(n - 1) \varepsilon^{2}}, & \text{when}\ 1 \leq \varepsilon < 6. \\
\end{cases}
\]
\end{thm}

\begin{proof}
Let $\vec{D} = (\vec{x}_{1}, \dots, \vec{x}_{n})$ and $\vec{D}' = (\vec{x}_{1}, \dots, \vec{x}'_{n})$ be the two neighbouring datasets differing only in the input of the $n$\textsuperscript{th} user, as used in Section~\ref{sec:basic}.
Here each vector-valued message $\vec{x}_{i}$ is of the form $(x_{i}^{(1)}, \dots, x_{i}^{(d)})$.
Recall from Section~\ref{sec:basic} that we assume that the analyzer can see the users in $B$ (i.e., the subset of users that returned a uniformly random message), as well as the inputs from the first $n - 1$ users.

We now introduce the \emph{vector view} $\text{VView}_{\mathcal{M}}(\vec{D})$ as the collection of information that the analyzer is able to see after the mechanism $\mathcal{M}$ is applied to all vector-valued messages in the dataset $\vec{D}$.
$\text{VView}_{\mathcal{M}}(\vec{D})$ is defined as the tuple 
$(\vec{Y}, \vec{D}_{\cap}, \vec{b})$, where $\vec{Y}$ is the multiset containing the outputs $\{ \vec{y}_{1}, \dots, \vec{y}_{n} \}$ of the mechanism $\mathcal{M}(\vec{D})$, $\vec{D}_{\cap}$ is the vector containing the inputs $(\vec{x}_{1}, \dots, \vec{x}_{n - 1})$ from the first $n - 1$ users, and $\vec{b}$ contains binary vectors $(\vec{b}_{1}, \dots, \vec{b}_{n})$ which indicate for which coordinates each user reports truthful information.
This vector view can be projected to $t$ overlapping \emph{scalar views} by applying Algorithm~\ref{alg:fixedpoint} only to the $j^{\text{th}}$ uniformly sampled coordinate $\alpha_{j} \in [d]$ from each user, where $j \in [t]$.
The $j^{\text{th}}$ scalar view $\text{View}_{\mathcal{M}}^{(\alpha_{j})}(\vec{D})$ of $\text{VView}_{\mathcal{M}}(\vec{D})$ is defined as the tuple $(\vec{Y}^{(\alpha_{j})}, \vec{D}_{\cap}^{(\alpha_{j})},  \vec{b}^{(\alpha_{j})})$, where:
\begin{align*}
\vec{Y}^{(\alpha_{j})} & = \mathcal{M} (\vec{D}^{(\alpha_{j})}) = \{ y_{1}^{(\alpha_{j})}, \dots, y_{n}^{(\alpha_{j})} \}, \\ \vec{D}_{\cap}^{(\alpha_{j})} & = (x_{1}^{(\alpha_{j})}, \dots, x_{n - 1}^{(\alpha_{j})}) \\
\text{and} \quad \vec{b}^{(\alpha_{j})} & = (b_{1}^{(\alpha_{j})}, \dots, b_{n}^{(\alpha_{j})})
\end{align*}
are the analogous definitions of $\vec{Y}$, $\vec{D}_{\cap}$ and $\vec{b}$, but containing only the information referring to the $j^{\text{th}}$ uniformly sampled coordinate of each vector-valued message.

The following \emph{advanced composition} results will be used in our setting to get a tight upper bound:

\begin{thm}[Dwork \emph{et al.}~\cite{dwork}] \label{thm:advcomp}
For all $\varepsilon', \delta', \delta \geq 0$, the class of $(\varepsilon', \delta')$-differentially private mechanisms satisfies $(\varepsilon, r \delta' + \delta)$-differential privacy under $r$-fold adaptive composition for:
\[
\varepsilon = \sqrt{2r \log(1/\delta)} \varepsilon' + r \varepsilon' \dftbig( e^{\varepsilon'} - 1 \dftbig).
\]
\end{thm}

\begin{cor} \label{cor:advcomp}
Given target privacy parameters $0 < \varepsilon < 1$ and $\delta > 0$, to ensure $(\varepsilon, r \delta' + \delta)$ cumulative privacy loss over $r$ mechanisms, it suffices that each mechanism is $(\varepsilon', \delta')$-DP, where:
\[
\varepsilon' = \frac{\varepsilon}{2 \sqrt{2r \log(1/\delta)}}.
\]
\end{cor}

\noindent To show that $\text{VView}_{\mathcal{M}}(\vec{D})$ satisfies $(\varepsilon, \delta)$-DP it suffices to prove that:
\[
\Pr_{\widetilde{\mathsf{V}} \sim \text{VView}_{\mathcal{M}}(\vec{D})} \!\left[ \frac{ \Pr [ \text{VView}_{\mathcal{M}}(\vec{D}) = \widetilde{\mathsf{V}} ] }
{ \Pr [ \text{VView}_{\mathcal{M}}(\vec{D}') = \widetilde{\mathsf{V}} ] } \geq e^{\varepsilon} \right] \leq \delta. \label{eq:vector} \tag{$1$}
\]

\noindent By considering this vector view as a union of overlapping scalar views, and letting $r = t$ in Corollary~\ref{cor:advcomp}, it is sufficient to derive \eqref{eq:vector} from:
\[
\Pr_{\mathsf{V}_{\alpha_{j}} \sim \text{View}_{\mathcal{M}}^{(\alpha_{j})}(\vec{D})} \!\left[ \frac{ \Pr [ \text{View}_{\mathcal{M}}^{(\alpha_{j})}(\vec{D}) = \mathsf{V}_{\alpha_{j}} ] }
{ \Pr [ \text{View}_{\mathcal{M}}^{(\alpha_{j})}(\vec{D}') = \mathsf{V}_{\alpha_{j}} ] } \geq e^{\varepsilon'} \right] \leq \delta', \label{eq:scalar} \tag{$2$}
\]

\noindent where $\widetilde{\mathsf{V}} = \bigcup_{\alpha_{j}} \mathsf{V}_{\alpha_{j}}$, $\varepsilon' = \frac{\varepsilon}{2 \sqrt{2t \log(1/\delta)}}$ and $\delta' = \frac{\delta}{t}$.

\begin{lem} \label{lem:implies}
Condition \eqref{eq:scalar} implies condition \eqref{eq:vector}.
\end{lem}

\begin{proof}
We can express $\text{VView}_{\mathcal{M}}(\vec{D})$ as the {composition} of the $t$ scalar views $\text{View}_{\mathcal{M}}^{(\alpha_{1})}, \dots, \text{View}_{\mathcal{M}}^{(\alpha_{t})}$, as:
\begin{align*}
\Pr&[ \text{VView}_{\mathcal{M}}(\vec{D}) = \widetilde{\mathsf{V}} ] \\
&= \Pr [ \text{View}_{\mathcal{M}}^{(\alpha_{1})}(\vec{D}) = \mathsf{V}_{\alpha_{1}} \wedge \cdots \wedge \text{View}_{\mathcal{M}}^{(\alpha_{t})}(\vec{D}) = \mathsf{V}_{\alpha_{t}} ] \\
&= \Pr [ \text{View}_{\mathcal{M}}^{(\alpha_{1})}(\vec{D}) = \mathsf{V}_{\alpha_{1}} ] \boldsymbol{\cdot} \cdots \boldsymbol{\cdot} \Pr [ \text{View}_{\mathcal{M}}^{(\alpha_{t})}(\vec{D}) = \mathsf{V}_{\alpha_{t}}].
\end{align*}

Our desired result is immediate by applying Corollary~\ref{cor:advcomp}, which states that the use of $t$ overlapping $(\varepsilon', \delta')$-DP mechanisms, when taken together, is $(\varepsilon, \delta)$-DP.
This applies in our setting, since we have assumed that $\text{VView}_{\mathcal{M}}(\vec{D})$ satisfies the requirements of $(\varepsilon, \delta)$-DP, and that each of the $t$ overlapping scalar views is formed identically but for a different uniformly sampled coordinate of the vector-valued messages.
\end{proof}

To complete the proof of Theorem~\ref{thm:main} for $\varepsilon < 1$, it remains to show that for a uniformly sampled coordinate $\alpha_{j} \in [d]$, $\text{View}_{\mathcal{M}}^{(\alpha_{j})}(\vec{D})$ satisfies $(\varepsilon', \delta')$-DP.

\begin{restatable}{lem}{primelemma} \label{lem:scalar}
Condition \eqref{eq:scalar} holds.
\end{restatable}

\begin{proof}
See Appendix~\ref{app:proof}.
\end{proof}

We now show that the above proof can be adjusted to cover the additional case $1 \leq \varepsilon < 6$.
This will be sufficient to complete the proof of our main Theorem~\ref{thm:main}.

First, we scale the setting of $\varepsilon'$ by a multiple of $6$ in Corollary~\ref{cor:advcomp} so that the advanced composition property holds for all $1 \leq \varepsilon < 6$.
We now insert $\varepsilon' = \frac{\varepsilon}{12 \sqrt{2r \log(1/\delta)}}$ into the proof of Theorem~\ref{thm:main}, resulting in a change of constant from $56$ to $2016$.
\end{proof}

\subsection{Accuracy Bounds for Shuffled Vector Sum} \label{sec:tightupperbound}

We now formulate an upper bound for the MSE of our protocol, and then identify the value(s) of $t$ that minimize this upper bound.

First, note that encoding the coordinate $x_{i}^{(\alpha_{j})}$ as $\bar{x}_{i}^{(\alpha_{j})} = \floor[\dftbig]{x_{i}^{(\alpha_{j})} k} \ +$ {\tt Ber}$ ( x_{i}^{(\alpha_{j})} k - \floor[\dftbig]{x_{i}^{(\alpha_{j})} k})$ in Algorithm~\ref{alg:fixedpoint} ensures that $\mathbb{E}[\bar{x}_{i}^{(\alpha_{j})}/k] = \mathbb{E} [x_{i}^{(\alpha_{j})}]$. This means that our protocol is unbiased.
For any unbiased random variable $X$ with $a < X < b$ then $\text{Var}[X] \leq (b-a)^2/4$, and so the MSE per coordinate due to the fixed-point approximation of the true vector in $\mathcal{R}_{d, k, n, t}$  is at most $\frac{1}{4k^{2}}$. 
Meanwhile, the MSE when $\mathcal{R}_{d, k, n, t}$ submits a random vector is at most $\frac{1}{2}$ per coordinate.

We now use the unbiasedness of our protocol to obtain a result for estimating the squared error between the estimated average vector and the true average vector.
When calculating the MSE, each coordinate location is used with expectation $n/d$.
Therefore, we define the \emph{normalized} MSE, or $\widehat{\textnormal{MSE}}$, as the normalization of the MSE by a factor of $(n/d)^2$.

\begin{thm} \label{thm:mse}
For any $d, n \in \mathbb{N}$, $\{ t \in \mathbb{N} \ | \ t \in [d] \}$, $\varepsilon < 6$ and $\delta \in \!\left( 0,1 \right]$, there exists a parameter $k$ such that $\mathcal{P}_{d, k, n, t}$ is $(\varepsilon, \delta)$-DP and
\[
\widehat{\textnormal{MSE}}(\mathcal{P}_{d, k, n, t}) =
\begin{cases}
\frac{2t d^{8/3} (14 \log(1/\delta) \log(2t/\delta))^{2/3}}{(1-\gamma)^{2} n^{5/3} \varepsilon^{4/3}}, \\ \quad \text{when}\ \varepsilon < 1 \\
\frac{8t d^{8/3} (63 \log(1/\delta) \log(2t/\delta))^{2/3}}{(1-\gamma)^{2} n^{5/3} \varepsilon^{4/3}}, \\ \quad \text{when}\ 1 \leq \varepsilon < 6, \\
\end{cases}
\]
where $\widehat{\textnormal{MSE}}$ denotes the squared error between the estimated average vector and the true average vector.
\end{thm}

\begin{proof}
We consider the $\sum_{l = 1}^{d} {\tt DeBias}(\hat{z}^{(l)})$ of $\mathcal{P}_{d, k, n, t}$ compared to the corresponding input $\sum_{j = 1}^{t} \sum_{i = 1}^{n} x_{i}^{(\alpha_{j})}$ over the dataset $\vec{D}$.
We use the bounds on the variance of the randomized response mechanism from Theorem~\ref{thm:mse} to give us an upper bound for this comparison.
\begin{align*}
\textnormal{MSE}&(\mathcal{P}_{d, k, n, t}) = \sup_{\vec{D}} \E \!\left[ \!\left( \sum_{l = 1}^{d} {\tt DeBias} ( \hat{z}^{(l)} ) - \sum_{j = 1}^{t} \sum_{i = 1}^{n} x_{i}^{(\alpha_{j})} \right)^{2} \ \right] \\
&= \sup_{\vec{D}} \E \!\left[ \!\left( \sum_{j = 1}^{t} \sum_{i = 1}^{n} \!\left( {\tt DeBias} ( y_{i}^{(\alpha_{j})}/k ) - x_{i}^{(\alpha_{j})} \right) \right)^{2} \ \right] \\
&= \sup_{\vec{D}} \sum_{j = 1}^{t} \sum_{i = 1}^{n} \E \!\left[ \!\left( {\tt DeBias} ( y_{i}^{(\alpha_{j})}/k ) - x_{i}^{(\alpha_{j})} \right)^{2} \ \right] \\
&= \sup_{\vec{D}} \sum_{j = 1}^{t} \sum_{i = 1}^{n} \text{Var} \!\left[ {\tt DeBias} ( y_{i}^{(\alpha_{j})}/k ) \right] \\
&= \frac{tn}{(1 - \gamma)^{2}} \sup_{x_{1}^{(\alpha_{1})}} \text{Var} [ y_{1}^{(\alpha_{1})}/k ]
\le \frac{tn}{(1 - \gamma)^{2}} \!\left( \frac{1 - \gamma}{4k^{2}} + \frac{\gamma}{2} \right) \\
&\le \frac{tn}{(1 - \gamma)^{2}} \!\left( \frac{1}{4k^{2}} + \frac{A_{\varepsilon} dk  \log(1/\delta)\log(2t/\delta)}{(n-1) \varepsilon^2} \right),
\end{align*}

\noindent where $A_{\varepsilon} = 28$ when $\varepsilon < 1$, and $A_{\varepsilon} = 1008$ when $1 \leq \varepsilon < 6$.
In other words, $A_{\varepsilon}$ is equal to half the constant term in the expression of $\gamma$ stated in Theorem~\ref{thm:main}.
The choice
$k = \frac{(n-1) \varepsilon^{2}}{4A_{\varepsilon} d \log(1/\delta) \log(2t/\delta)}$ minimizes the bracketed sum above and the bounds in the statement of the theorem follow.
\end{proof}

To obtain the error between the estimated average vector and the true average vector, we simply take the square root of the result obtained in Theorem~\ref{thm:mse}.

\begin{cor} \label{cor:query}
For every statistical query $q: \mathcal{X} \mapsto [0,1]^{d}$, $d, n \in \mathbb{N}$, $\{ t \in \mathbb{N} \ | \ t \in [d] \}$, $\varepsilon < 6$ and $\delta \in \!\left( 0,1 \right]$, there is an $(\varepsilon, \delta)$-DP $n$-party unbiased protocol for estimating $\frac{d}{n} \sum_{i} q(\vec{x}_{i})$ in the Single-Message Shuffle Model with standard deviation
\[
\hat{\sigma}(\mathcal{P}_{d, k, n, t}) =
\begin{cases}
\frac{(2t)^{1/2} d^{4/3} (14 \log(1/\delta) \log(2t/\delta))^{1/3}}{(1-\gamma) n^{5/6} \varepsilon^{2/3}}, \\ \quad \text{when}\ \varepsilon < 1 \\
\frac{(8t)^{1/2} d^{4/3} (63 \log(1/\delta) \log(2t/\delta))^{1/3}}{(1-\gamma) n^{5/6} \varepsilon^{2/3}}, \\ \quad \text{when}\ 1 \leq \varepsilon < 6, \\
\end{cases}
\]
where $\hat{\sigma}$ denotes the error between the estimated average vector and the true average vector.
\end{cor}

To summarize, we have produced an unbiased protocol for the computation of the sum of $n$ real vectors in the Single-Message Shuffle Model with normalized MSE $O_{\varepsilon, \delta} (d^{8/3} t n^{-5/3})$, using advanced composition results from Dwork \emph{et al.}~\cite{dwork}.
Minimizing this bound as a function of $t$ leads us to choose $t=1$, but any choice of $t$ that is small and not dependent on $d$ produces a bound of the same order.
In our experimental study, we determine that the best choice of $t$ in practice is indeed $t=1$.

\subsection{Improved bounds for t=1} \label{sec:betterbounds}
We observe that in the optimal case in which $t=1$, we can tighten the bounds further, as we do not need to invoke the advanced composition results when each user samples only a single coordinate.
This changes the value of $\gamma$ by a factor of $O(\log(1/\delta))$, which propagates through to the expression for the MSE.
That is, we can more simply set $\varepsilon' = \varepsilon$ and $\delta' = \delta$ in the proof of Theorem~\ref{thm:main}. When $\varepsilon < 1$, the computation is straightforward, with $c \geq \frac{14}{\varepsilon'^2} \log(2t/\delta)$ being chosen as before. However, when $1 \leq \varepsilon < 6$, a tighter $c \geq \frac{80}{\varepsilon'^2} \log(2t/\delta)$ must be selected, as the condition $\varepsilon' < 1$ no longer holds.

Using $\varepsilon' < 6$, we have:
\[
(1 - \exp \hspace{.5mm} (-\varepsilon'/2)) \ge \left( 1 - \exp \left (-\frac{2}{3\sqrt{15}} \right) \right) \varepsilon' \ge \frac{\varepsilon'}{2\sqrt{10}}.
\]

\noindent Thus, we have:
\begin{align*}
\Pr \!\left[ \frac{\mathsf{N}_{\theta}}{\mathsf{N}_{\phi}}  \geq e^{\varepsilon'} \right]
&\le \exp \dftbig( -\frac{c}{3} (\varepsilon'/2)^2 \dftbig) +
\exp \dftbig( -\frac{c}{2} \dftbig( \frac{\varepsilon'}{2\sqrt{10}} \dftbig) ^2 \dftbig)\\
&\le 2\exp\left(-\frac{80}{2\varepsilon'^2}\frac{\varepsilon'^2}{40} \log(2t/\delta)\right) \leq \delta/t,
\end{align*}

\noindent which yields:
\[
\gamma =
\begin{cases}
\max \!\dftbig\{ \frac{14dk \log(2/\delta)}{(n - 1) \varepsilon^{2}}, \frac{27dk}{(n - 1) \varepsilon} \dftbig\}, & \text{when}\ \varepsilon < 1 \\
\max \!\dftbig\{ \frac{80dk \log(2/\delta)}{(n - 1) \varepsilon^{2}}, \frac{36dk}{11(n - 1) \varepsilon} \dftbig\}, & \text{when}\ 1 \leq \varepsilon < 6. \\
\end{cases}
\]

Note that the above expression for $\gamma$ in the case $\varepsilon < 1$ coincides with the result obtained by Balle \emph{et al.} in the scalar case~\cite{balleprivacyblanket}. Putting this expression for $\gamma$ in the proof of Theorem~\ref{thm:mse}, with the choice
\[
k = 
\begin{cases}
\min \!\dftbig\{ \!\left( \frac{n \varepsilon^{2}}{28d \log(2/\delta)} \right)^{1/3}, \ \!\left( \frac{n \varepsilon}{54d} \right)^{1/3} \dftbig\}, & \text{when}\ \varepsilon < 1 \\
\min \!\dftbig\{ \!\left( \frac{n \varepsilon^{2}}{160d \log(2/\delta)} \right)^{1/3}, \ \!\left( \frac{11n \varepsilon}{72d} \right)^{1/3} \dftbig\}, & \text{when}\ 1 \leq \varepsilon < 6,
\\
\end{cases}
\]

\noindent causes the upper bound on the normalized MSE to reduce to:
\[
\widehat{\textnormal{MSE}} =
\begin{cases}
\max \!\dftbig\{ \frac{98^{1/3} d^{8/3} \log^{2/3}(2/\delta)}{(1-\gamma)^{2} n^{5/3} \varepsilon^{4/3}}, \frac{18 d^{8/3}}{(1-\gamma)^{2} n^{5/3} (4 \varepsilon)^{2/3}} \dftbig\}, & \text{when}\ \varepsilon < 1 \\
\max \!\dftbig\{ \frac{2 d^{8/3} (20\log(2/\delta))^{2/3}}{(1-\gamma)^{2} n^{5/3} \varepsilon^{4/3}}, \frac{2(9^{2/3}) d^{8/3}}{(1-\gamma)^{2} n^{5/3} (11 \varepsilon)^{2/3}} \dftbig\}, & \text{when}\ 1 \leq \varepsilon < 6. \\
\end{cases}
\]

By updating Corollary \ref{cor:query} in the same way, we can conclude that for the optimal choice $t = 1$, the normalized standard deviation of our unbiased protocol can be further tightened to:
\[
\hat{\sigma} =
\begin{cases}
\max \!\dftbig\{ \frac{98^{1/6} d^{4/3} \log^{1/3}(2/\delta)}{(1-\gamma) n^{5/6} \varepsilon^{2/3}}, \frac{18^{1/2} d^{4/3}}{(1-\gamma) n^{5/6} (4 \varepsilon)^{1/3}} \dftbig\}, & \text{when}\ \varepsilon < 1 \\
\max \!\dftbig\{ \frac{2^{1/2} d^{4/3} (20\log(2/\delta))^{1/3}}{(1-\gamma) n^{5/6} \varepsilon^{2/3}}, \frac{2^{1/2} 9^{1/3} d^{4/3}}{(1-\gamma) n^{5/6} (11 \varepsilon)^{1/3}} \dftbig\}, & \text{when}\ 1 \leq \varepsilon < 6. \\
\end{cases}
\]

\section{Experimental Evaluation} \label{sec:eeval}

In this section we present and compare the bounds generated by applying Algorithms~\ref{alg:fixedpoint} and \ref{alg:analyzer} to an ECG Heartbeat Categorization Dataset in Python, available at \url{https://www.kaggle.com/shayanfazeli/heartbeat}.
We analyse the effect of changing one key parameter at a time, whilst the others remain the same.
Our default settings are vector dimension $d = 100$, rounding parameter $k = 3$, number of users $n = 50000$, number of coordinates to sample $t = 1$, and differential privacy parameters $\varepsilon = 0.95$ and $\delta = 0.5$.
The ranges of all parameters have been adjusted to best display the dependencies, whilst simultaneously ensuring that the parameter $\gamma$ of the randomized response mechanism is always within its permitted range of $[0,1]$.
The Python code is available at \url{https://github.com/mary-python/dft/blob/master/shuffle}.

We first confirm that the choice of $t = 1$ is optimal, as predicted by the results of Section~\ref{sec:betterbounds}.
Indeed, Fig.~\ref{fig:basictk} (a) shows that the total experimental $\widehat{\textnormal{MSE}}$ for the ECG Heartbeat Categorization Dataset is significantly smaller when $t = 1$, compared to any other small value of $t$.

\begin{figure*}[p]
\centering
\subfloat[Experimental error by number of coordinates $t$ retained]{\includegraphics[width=0.4\linewidth]{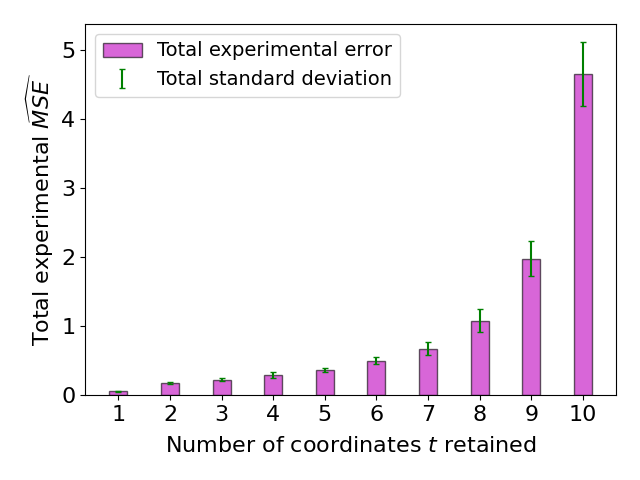}}%
\qquad\qquad
\subfloat[Experimental error by number of buckets $k$ used]{\includegraphics[width=0.4\linewidth]{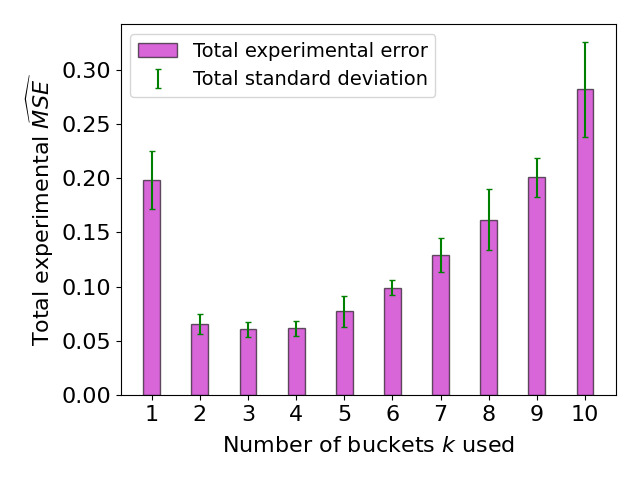}}\\
\caption{\small\label{fig:basictk} Bar charts confirming that the choices $t = 1$ (a) and $k = 3$ (b) minimize the total experimental $\widehat{\textnormal{MSE}}$ for the ECG Heartbeat Categorization Dataset.}
\end{figure*}

\begin{figure*}[p]
\centering
\subfloat[Experimental error by vector dimension $d$]{\includegraphics[width=0.4\linewidth]{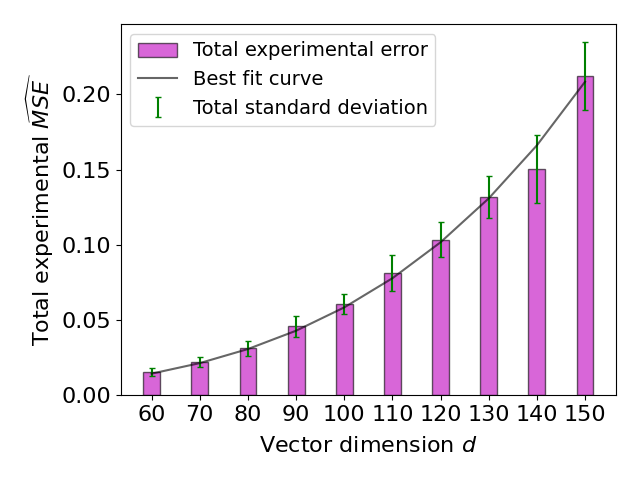}}%
\qquad\qquad
\subfloat[Experimental error by number of vectors $n$ used]{\includegraphics[width=0.4\linewidth]{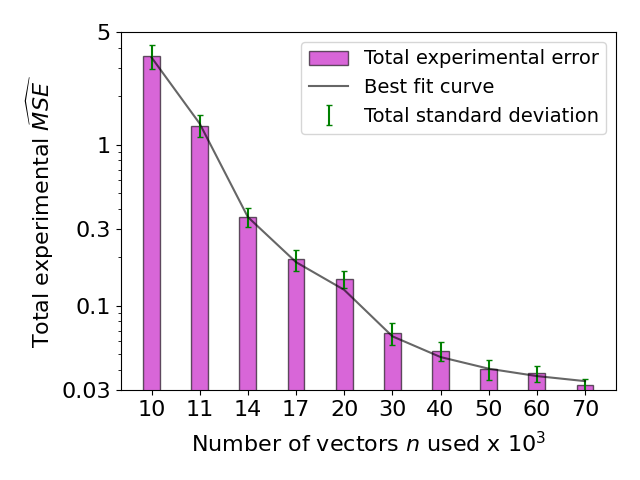}}\\
\caption{\small\label{fig:basicdn} Bar charts with best fit curves confirming the dependencies $d^{8/3}$ (a) and $n^{-5/3}$ (b) from Theorem~\ref{thm:mse}.}
\end{figure*}

\begin{figure*}[p]
\centering
\subfloat[Experimental error by value of $\varepsilon$ where $\varepsilon < 1$]{\includegraphics[width=0.4\linewidth]{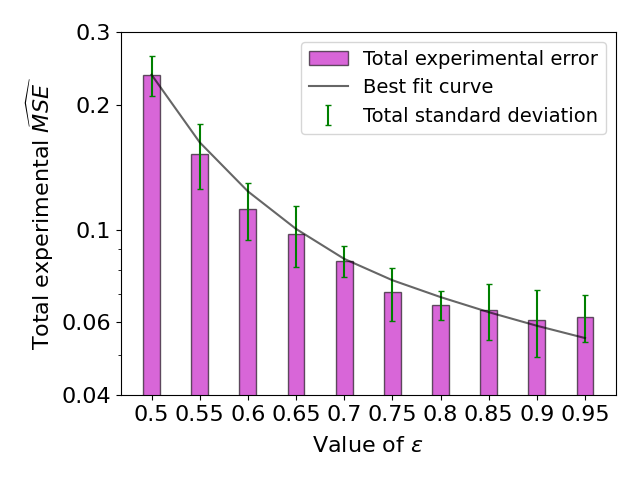}}%
\qquad\qquad
\subfloat[Experimental error by value of $\varepsilon$ where $1 \leq \varepsilon < 6$]{\includegraphics[width=0.4\linewidth]{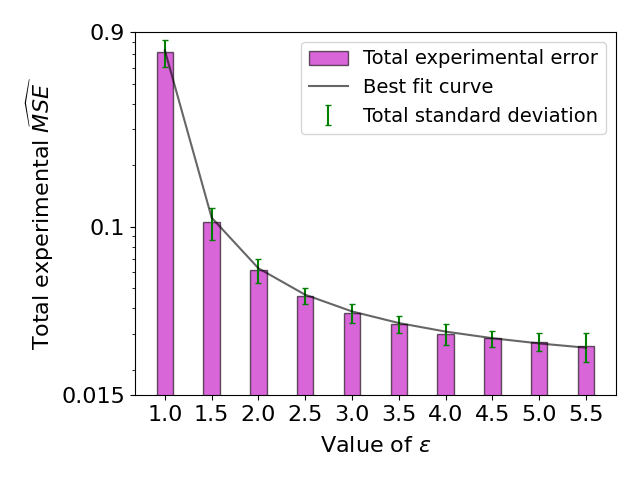}}\\
\caption{\small\label{fig:basiceps} Bar charts with best fit curves confirming the dependency $\varepsilon^{-4/3}$ from Theorem~\ref{thm:mse} in the two ranges $\varepsilon < 1$ (a) and $1 \leq \varepsilon < 6$ (b).}
\end{figure*}

Similarly, Fig.~\ref{fig:basictk} (b) suggests that the total experimental $\widehat{\textnormal{MSE}}$ is lowest when $k = 3$, which is sufficiently close to the choice of $k$ selected in the proof of Theorem~\ref{thm:mse}, with all other default parameter values substituted in.
Observe that the absolute value of the observed MSE is below 0.3 in this case, meaning that the vector is reconstructed to a high degree of accuracy, sufficient for many applications.

Next, we verify the bounds of $d^{8/3}$ and $n^{-5/3}$ from Theorem~\ref{thm:mse}.
Fig.~\ref{fig:basicdn} (a) is plotted with a best fit curve with equation a multiple of $d^{8/3}$, exactly as desired.
Unsurprisingly, the MSE increases as $d$ goes up according to this superlinear dependence.
Meanwhile, Fig.~\ref{fig:basicdn} (b) fits a curve dependent on $n^{-7/6}$, sufficiently close to the required result.
We see the benefit of increasing $n$: as $n$ increases by a factor of 10 across the plot, the error decreases by more than two orders of magnitude.
In Fig.~\ref{fig:basiceps}, we verify the dependency $\varepsilon^{-4/3}$ in the two ranges $\varepsilon < 1$ and $1 \leq \varepsilon < 6$.
The behavior for $\varepsilon < 1$ is quite smooth, but becomes more variable for larger $\varepsilon$ values.

In conclusion, these experiments confirm that picking $t = 1$ and $k = 3$ serves to minimize the error.
The lines of best fit confirm the dependencies on the other parameters from Section~\ref{sec:vectorsum} for $d$, $\varepsilon$ and $n$, by implementing and applying Algorithms~\ref{alg:fixedpoint} and \ref{alg:analyzer} to an ECG Heartbeat Categorization Dataset in Python.
The experiments demonstrate that the MSE observed in practice is sufficiently small to allow effective reconstruction of average vectors for a suitably large cohort of users.

\section{Conclusion} \label{sec:conc}

Our results extend a result from Balle \emph{et al.}~\cite{balleprivacyblanket} for scalar sums to provide a protocol $\mathcal{P}_{d, k, n, t}$ in the Single-Message Shuffle Model for the private summation of vector-valued messages $(\vec{x}_{1}, \dots, \vec{x}_{n}) \in ([0,1]^{d})^{n}$.
It is not surprising that the normalized MSE of the resulting estimator has a dependence on $n^{-5/3}$, as this was the case for scalars, but the addition of a new dimension $d$ introduces a new dependency for the bound, as well as the possibility of sampling $t$ coordinates from each $d$-dimensional vector.
For this extension, we formally defined the \emph{vector view} as the knowledge of the analyzer upon receiving the randomized vectors, and expressed it as a union of overlapping scalar views.
Through the use of advanced composition results from Dwork \emph{et al.}~\cite{dwork}, we showed that the estimator now has normalized MSE $O_{\varepsilon, \delta} (d^{8/3} t n^{-5/3})$ which can be further improved to $O_{\varepsilon, \delta} (d^{8/3} n^{-5/3})$ by setting $t = 1$.

Our contribution has provided a stepping stone between the summation of the scalar case discussed by Balle \emph{et al.}~\cite{balleprivacyblanket} and the linearization of more sophisticated structures such as matrices and higher-dimensional tensors, both of which are reliant on the functionality of the vector case.
As mentioned in Section~\ref{sec:litreview}, there is potential for further exploration in the Multi-Message Shuffle Model to gain additional privacy, echoing the follow-up paper of Balle \emph{et al.}~\cite{ballemulti}.

\appendix
\section{Proof of Lemma~\ref{lem:scalar}} \label{app:proof}
\primelemma*

\begin{proof}
The way in which we split the vector view (i.e., to consider a single uniformly sampled coordinate of each vector-valued message in turn), means that we can apply a proof that is analogous to the scalar-valued case~\cite{balleprivacyblanket}.
We work through the key steps needed.

Recall from Section~\ref{sec:basic} that the case where the $n^{\text{th}}$ user submits a uniformly random message independent of their input satisfies DP trivially.
Otherwise, the $n^{\text{th}}$ user submits their true message, and we assume that analyzer removes from $\vec{Y}^{(\alpha_{j})}$ any truthful messages associated with the first $n - 1$ users.
Denote $n_{l}^{(\alpha_{j})}$ to be the count of $j^{\text{th}}$ coordinates remaining with a particular value $l \in [k]$. If $\vec{x}_{n}^{(\alpha_{j})} = \theta$ and $\vec{x}_{n}^{\prime (\alpha_{j})} = \phi$, we obtain the relationship
\[
\frac{\Pr [ \text{View}_{\mathcal{M}}^{(\alpha_{j})}(\vec{D}) = V_{\alpha_{j}} ] }{\Pr [ \text{View}_{\mathcal{M}}^{(\alpha_{j})}(\vec{D}') 
= V_{\alpha_{j}} ] } = \frac{n_{\theta}^{(\alpha_{j})}}{n_{\phi}^{(\alpha_{j})}}.
\]

\noindent We observe that the counts $n_{\theta}^{(\alpha_{j})}$ and $n_{\phi}^{(\alpha_{j})}$ follow the binomial distributions $\mathsf{N}_\theta \sim {\tt Bin} \dftbig( s, \frac{\gamma}{k} \dftbig) + 1$ and $\mathsf{N}_\phi \sim {\tt Bin} \dftbig( s, \frac{\gamma}{k} \dftbig) $ respectively, where $s$ denotes the number of times that the coordinate $j$ is sampled.
In expectation, $s = (n-1)t/d$, and below we will show that it is close to its expectation:

\begin{align*}
&\Pr_{\mathsf{V}_{\alpha_{j}} \sim \text{View}_{\mathcal{M}}^{(\alpha_{j})}(\vec{D})} \!\left[ \frac{ \Pr [ \text{View}_{\mathcal{M}}^{(\alpha_{j})}(\vec{D}) = \mathsf{V}_{\alpha_{j}} ] }
{ \Pr [ \text{View}_{\mathcal{M}}^{(\alpha_{j})}(\vec{D}') = \mathsf{V}_{\alpha_{j}} ] } \geq e^{\varepsilon'} \right]\\
&= \Pr \!\left[ \frac{\mathsf{N}_{\theta}}{\mathsf{N}_{\phi}} \geq e^{\varepsilon'} \right].
\end{align*}

\noindent We define $c:= \E[\mathsf{N}_{\phi}] = \frac{\gamma}{k} \cdot s$ and split this into the union of two events, $N_\theta \geq c e^{\varepsilon'/2}$ and $N_\phi \leq c e^{-\varepsilon'/2}$.
Applying a Chernoff bound gives:
\begin{align*}
\Pr \!\left[ \frac{\mathsf{N}_{\theta}}{\mathsf{N}_{\phi}} \geq e^{\varepsilon'} \right] &\le \exp \!\left( - \frac{c}{3} \!\left( e^{\varepsilon'/2} - 1 - \frac{1}{c} \right)^{2} \right)\\
&+ \exp \!\left( - \frac{c}{2} \!\left( 1 - e^{- \varepsilon'/2} \right)^{2} \right).
\end{align*}

\noindent We will choose $c \geq \frac{14}{\varepsilon'^2} \log(2t/\delta)$ so that we have:
\[
\exp \hspace{.5mm} (\varepsilon'/2) - 1 - \frac{1}{c}
\ge \frac{\varepsilon'}{2} + \frac{\varepsilon'^2}{8} - \frac{\varepsilon'^2}{14\log(2t/\delta)}
\ge \frac{\varepsilon'}{2}.
\]

\noindent Using $\varepsilon' < 1$, we have:
\[
(1 - \exp \hspace{.5mm} (-\varepsilon'/2)) \ge (1 - \exp \hspace{.5mm} (-1/2))\varepsilon' \ge \frac{\varepsilon'}{\sqrt{7}}.
\]

\noindent Thus we have:
\begin{align*}
\Pr \!\left[ \frac{\mathsf{N}_{\theta}}{\mathsf{N}_{\phi}}  \geq e^{\varepsilon'} \right]
&\le \exp \dftbig( -\frac{c}{3} (\varepsilon'/2)^2 \dftbig) +
\exp \dftbig( -\frac{c}{2}(\varepsilon'/\sqrt{7})^2 \dftbig)\\
&\le 2\exp\left(-\frac{14}{2\varepsilon'^2}\frac{\varepsilon'^2}{7} \log(2t/\delta)\right) \leq \delta/t.
\end{align*}

We now apply another Chernoff bound to show that $s \leq 2\E[s]$, which can be used to give a bound on $\gamma$.
The following calculation proves that $\Pr[s \geq 2\E(s)] \leq \exp(-\E(s)/3)$,
using $\E(s) = (n-1)t/d$:
\[
\Pr[s \geq 2\E(s)]
\leq  \exp \dftbig( -\frac{n-1}{3}t/d \dftbig) \leq  \exp \dftbig( -\frac{n}{3} \dftbig) < \delta/3t,
\]
for all reasonable values of $\delta$.

Substituting these bounds on $s$ and $c$ into $\gamma s/k = c$ along with
$\varepsilon'=\frac{\varepsilon}{2\sqrt{2t \log (1/\delta)}}$
gives:
\begin{align*}
\gamma \geq & \frac{112kt \log(1/\delta) \log(2t/\delta)}{s \varepsilon^{2}} 
\geq \frac{56dk \log(1/\delta)\log(2t/\delta)}{(n-1) \varepsilon^{2}}.
\end{align*}
\end{proof}

\end{document}